\documentclass[english,letter,12pt,twosided]{article}

\usepackage{amsfonts}
\usepackage[left=2.3cm,right=2.0cm, bottom = 2.3cm, top=2.4cm]{geometry}
\usepackage[utf8]{inputenc}
\usepackage{babel}
\usepackage{slantsc}
\usepackage{array}
\usepackage{amsmath}
\usepackage{amsthm, mathtools}
\usepackage{amssymb}
\usepackage[pdftex]{color,graphicx}
\usepackage{enumerate}
\usepackage{mathtools}
\usepackage{bbm}
\usepackage{tikz}
\usepackage{subfigure}
\usepackage{pgfplots}
\usepackage{utopia}
\usepackage[labelsep=period]{caption}
\usepackage{hyperref}
\usepackage{setspace}
\usepackage[linesnumbered,ruled]{algorithm2e} %code

\usetikzlibrary{patterns,snakes}
\newtheorem{theorem}{Theorem}

\newtheorem{lemma}{Lemma}

\theoremstyle{definition}

\newtheorem{example}{Example}

\newtheorem{remark}{Remark}

\newcommand{\ignore}[1]{}

\newcommand{\R}{{\mathbb{R}}}

\newcommand{\Z}{{\mathbb Z}}

\newcommand{\G}{{\mathcal{G}}}

\newcommand{\eg}{{\it e.g.}}
\newcommand{\ie}{{\it i.e.}}
\newcommand{\etc}{{\it etc}}

\newcommand{\ones}{\mathbf 1}
\DeclareMathOperator{\tr}{\mathop{Trace}}

\DeclareMathOperator{\diag}{diag}

\definecolor{PennBlue}{RGB}{001,031,091}
\definecolor{PennRed}{RGB}{153,0,0}
\definecolor{NewBlue}{RGB}{001,031,110}
\definecolor{NewRed}{RGB}{200,0,0}
\hypersetup{
pdfborder = {0 0 0},
    colorlinks,
    citecolor=PennRed,
    filecolor=PennRed,
    linkcolor=PennRed,
    urlcolor=PennRed
}
\begin{document}

\title{\fontsize{21}{21}%
\selectfont {Structural Analysis and Optimal Design of Distributed System Throttlers}}
\author{\fontsize{13}{13} {Milad Siami} \thanks{{\footnotesize The author performed the work as a summer intern at Google, NYC.}} \thanks{{\footnotesize Massachusetts Institute of Technology (MIT), Cambridge, MA 02139, USA. Email:  {\tt\small siami@mit.edu}.}}  \hspace{0.5in}{Jo\"{e}lle Skaf}\thanks{{\footnotesize Google Inc., 76 9th Avenue, New York, NY 10011. Email: {\tt\small jskaf@google.com}.}}  \\
%EndAName
~}
\date{\vspace{-5ex}}
\maketitle

\begin{abstract}
\fontsize{13}{13}\selectfont \baselineskip0.5cm
  In this paper, we investigate the performance analysis and synthesis of distributed system throttlers (DST). A throttler is a mechanism that limits the flow rate of incoming metrics, \eg, byte per second, network bandwidth usage, capacity, traffic,  \etc. This can be used to protect a service's backend/clients from getting overloaded, or to reduce the effects of uncertainties in demand for shared services. We study performance deterioration of DSTs subject to demand uncertainty. We then consider network synthesis problems that aim to improve the performance of noisy DSTs via communication link modifications as well as server update cycle modifications.
\newline
\newline
\newline
\end{abstract}
\fontsize{12}{12}\selectfont

%%%%%%%%%%%%%%%%%%%%%%%%%%%%%%%%%%%%%%%%%%%

\thispagestyle{empty} \newpage \fontsize{12}{12}\selectfont\baselineskip %
0.60cm

\onehalfspacing
%%%%%%%%%%%%%%%%%%%%%%%%%%%%%%%%%%%%%%%%%%%%%%%%%%%
\section{Introduction}
\allowdisplaybreaks

 System throttling (also known as rate-limiting) aims to limit the total number of requests from all clients to a shared service and provide a harmonized and fair quota allocation among them (where the definition of fairness is application-specific). Examples of systems in need of throttling protection include cloud-based services and traffic management services. A number of works on rate-limiting systems and congestion control have been published in the recent literature \cite{Xia, Tan, Zhang, Johari, Gibbens, Kelly,Raghavan,Stanojevic}.

System throttlers can be classified into centralized and distributed types. In a centralized system throttler (CST), there is a single decision maker that sets the per-client limits according to aggregated metrics it receives from multiple servers, which in turn aggregate them from metrics reported by the clients. CSTs are designed based on a globally aggregated view of usage metrics.
On the other hand, a distributed system throttler (DST) does not have a centralized mechanism for setting per-client limits: it consists of multiple servers, each of which makes autonomous decisions and updates its own limit based on measurements it takes as well as local information. 

While the centralized approach has benefits, including consistency and ease of implementation and analysis, it also has drawbacks relative to a decentralized version: (i) Less local adaptability: in a centralized version, each server needs to send information to the decision-making server and wait for its command, which means a delayed response time. (ii) Limited communication: there is no inter-server communication except to the decision-making server. Moreover, we want to facilitate information propagation in order to improve the performance and to make the network more flexible and fast when handling uncertainty in demand. %Demands are likely to shift and have unpredictable behavior, our experiments/simulations show that centralized methods do not perform well in these situations. Related work \cite{SiamiTAC,Raghavan,Stanojevic,siami}.

There are some related works in the literature that study performance and robustness issues in noisy linear
distributed systems; for example, see \cite{Bamieh12, SiamiTCNS, SiamiTAC, Zelazo-Mesbahi,Siami14cdc-2, LovisariGarinZampieriResistance,Lin2013, Siami13cdc, Alex} and the references therein. In \cite{Bamieh12}, the authors investigate the deviation from the mean of states of a continuous-time consensus network on tori with additive noise inputs. A rather comprehensive performance analysis of noisy linear consensus networks with arbitrary graph topologies has been recently reported in \cite{SiamiTAC}. In \cite{SiamiTAC}, several fundamental tradeoffs between a $\mathcal H_2$-based performance measure and sparsity measures of a continuous-time linear consensus network are studied. Moreover, \cite{Sharaf} studies a $\mathcal H_2$-based performance measure of continuous-time linear consensus system in the presence of a time-delay and additive noise inputs. Most of these papers treat continuous-time systems only; in discrete-time networks, however, the time-step size along with the topology of the network plays an important role on the network performance. 

We should mention that papers \cite{Raghavan} and \cite{Stanojevic} investigate the notion of distributed rate-limiting as a mechanism that controls the aggregate service used by a client of a cloud-based service. The main idea is to improve a set of cloud servers with the ability to exchange information with them towards the common purpose: control of the aggregate usage that a cloud-based service experiences. However, comprehensive performance analysis and synthesis have yet to be done for these networks with an arbitrary underlying graph. 

%{\it Our Contributions:}
In this paper, our goal is to develop a unified framework for analysis and design of discrete-time distributed rate-limiting systems with a local aggregated view of usage metrics. 
We investigate performance deterioration (\eg, over-throttling, mismatch, convergence rate) of DSTs with respect to external uncertainties and the update cycle of servers. We develop a graph-theoretic framework to relate the underlying structure of the system to its overall performance measure. {We then compare the performance/robustness of DSTs with different topologies. }  
In this work, in addition to the overall performance measure for a network, each node has its own performance measure, which is one of the main differentiators between this work and some other related work \cite{Bamieh12, SiamiTCNS, SiamiTAC,Sharaf}. 

{The rest of this paper is organized as follows. In Section \ref{sec:0}, we present some basic mathematical concepts and notations employed in this paper. In Section \ref{sec:1}, we define and study a distributed system throttler (DST). In Section \ref{sec:2}, we evaluate the overall performance of a DST with a given nodal performance measure. In Subsection \ref{sec:impact}, we study the impact of the server update cycle on performance. In Subsection \ref{sec:synthesis}, two synthesis problems are studied. In Section \ref{sec:numerical}, some numerical results are demonstrated. In Section \ref{sec:algorithms}, we focus on throttling algorithms which are used by servers. In Section \ref{sec:conclusion}, we conclude our work and suggest directions for future research. }

%%%%%%%%%%%%%%%%%%%%%%%%%%%%%%%%%%%%%%%%%%%%%%%%%%
\section{Mathematical Notation}
\label{sec:0}
%\subsection{Spectral Graph Theory}
Throughout the paper, 
the discrete time index is denoted by $k$. The sets of real (integer), positive real (integer), and strictly positive real (integer) numbers are represented by $\R$ ($\mathbb Z$), $\R_+$ ($\mathbb Z_+$) and $\R_{++}$ ($\mathbb Z_{++}$), respectively. Capital letters, such as $A$ or $B$, stand for real-valued matrices. We use $\diag(x_1, x_2, \ldots, x_n)$ to denote a $n$-by-$n$ diagonal square matrix with $x_1$ to $x_n$ on its diagonal. For a square matrix $X$, $\tr(X)$ refers to the
summation of on-diagonal elements of $X$.  We represent the $n$-by-$1$ vector of ones by $\ones$.  The $n$-by-$n$ identity matrix is denoted by $I$. The Moore-Penrose pseudoinverse of matrix $A$ is denoted by $A^{\dag}$, \ie, $A^{\dag}=\left(A + \frac{1}{n}\ones\ones^T\right)^{-1}-\frac{1}{n}\ones\ones^T$. 
We assume that all graphs are connected, undirected, simple graphs.
We represent graph $\mathcal G$ by $(V, E, w)$, where $V$ is the node set, $E$ is the edge set, and $w:E \rightarrow \mathbb R_+$ is the link weight function. 
We denote by $L$ the Laplacian matrix of the weighted graph $\mathcal G$ with the following eigenvalues
\begin{equation}
  \lambda_1=0 \leq \lambda_2 \leq \cdots \leq \lambda_n.
\end{equation}
Since we assume in this work that all graphs are connected, it follows that $\lambda_2 > 0$.

The {\it effective resistance} between nodes $i$ and $j$ is defined by:
	\begin{equation}
		r_{ij} ~:=~ l_{ii}^{\dag}+l_{jj}^{\dag}-l_{ji}^{\dag}-l_{ij}^{\dag}
	\label{eff-resis}
	\end{equation}
where $l_{ji}^{\dag}$ is the $(i,j)$th entry in $L^{\dag}$. 
The white Gaussian noise with zero mean and variance $\sigma^2$ is represented by $v \sim N(0,\sigma^2)$.

%%%%%%%%%%%%%%%%%%%%%%%%%%%%%%%%%%%%%%%%%%%%%%%%%%
\section{A Distributed System Throttler}
\label{sec:1}
\allowdisplaybreaks

A distributed system throttler (DST) is a graph $\mathcal G$ with $n$ nodes. Each node in the graph is a server with assigned clients that can send it requests. Links in the graph represent communication channels between servers. {The global goal of a DST is to keep the aggregate number of accepted requests from all clients for a shared service at or below a prescribed level.
The DST does not have a centralized mechanism for setting per-client limits. Instead it consists of multiple servers, each of which makes its own decisions and updates its own limit based on its own measurements and local information from its neighbors (on graph $\G$). Fig. \ref{Figure:DST} depicts an example of a distributed throttler with six nodes (servers).}

\begin{figure}
  \centering
  \includegraphics[trim = 0 0 0 0, clip,width=0.5\textwidth]{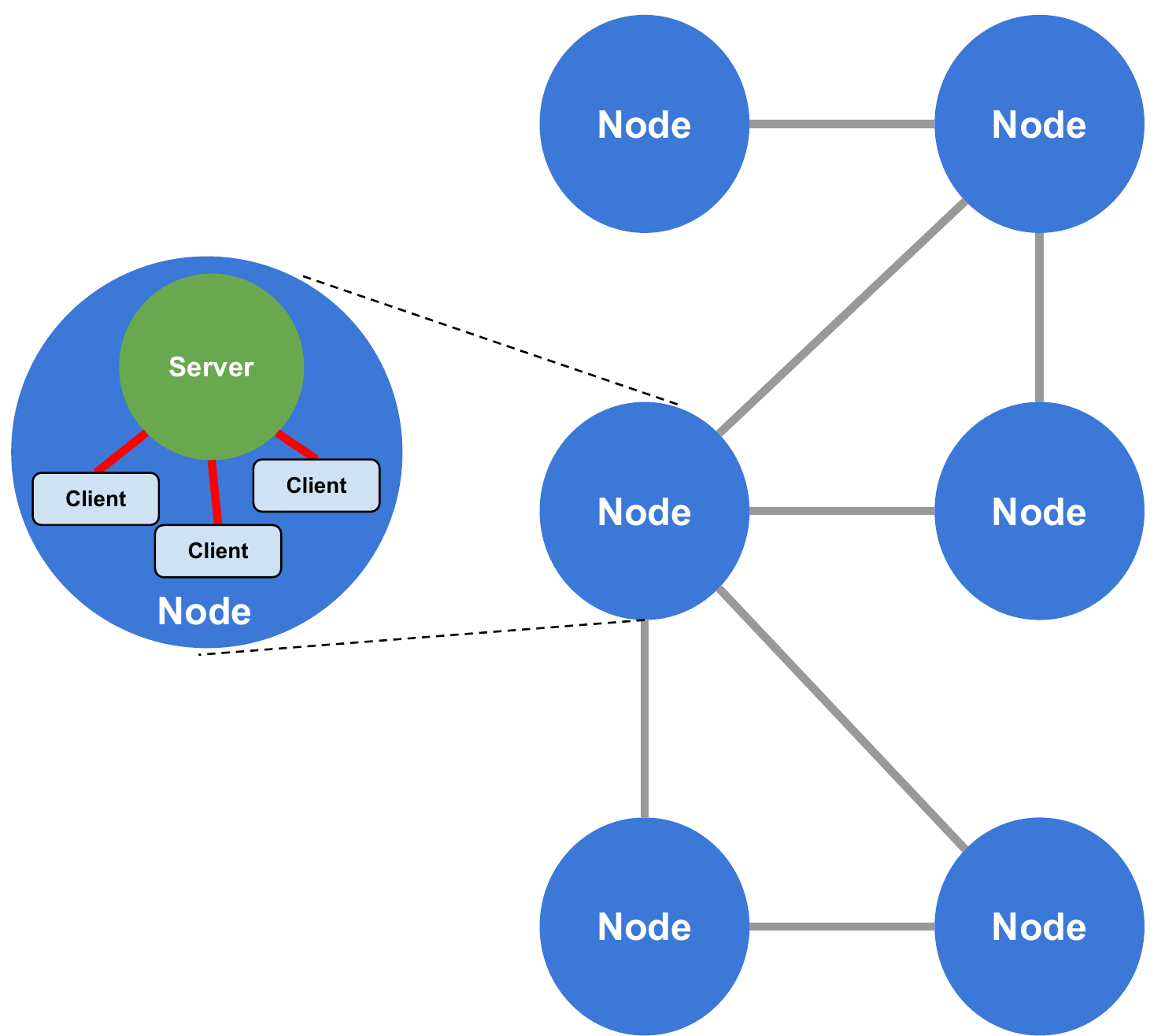}
  \caption{An example of a  distributed system throttler (DST) with 6 servers. Nodes show servers and links present communication links between servers. }
  \label{Figure:DST}
\end{figure}

Let's denote by $r_i(k)$ the total number of client requests received by server $i$ at time $k$. Each node has a total limit on the number of requests that it is allowed to service at time $k$ represented by $x_i(k)$. It is also associated with a performance measure $p_i(k)$ which represents how well that node is working at time $k$. Examples of typical performance measures are: over-throttling at time $k$, the ratio of the total allowed usage to total requested usage, or any function of $r_i(k)$, $x_i(k)$, and time. We will talk about functional properties of the performance measure later on in this paper (see Table I). 
%At the microscopic level, each node works exactly like a global throttler; each node i has a total limit at time $k$ and we present it by $x_i(k)$ and we denote the total user requests on server i by $r_i(k)$. Also each %node has a performance measure $p_i(k)$ which shows how well that node works at time $k$. This measure can be considered as over-throttling at time $k$, the average over-throttling, the ratio of the total assigned %usage to total requested usage, or any function of total usage, assigned usage, and time. We will talk about functional properties of the performance measure later on in this paper (please see Table I).    
%
\begin{table*}[h]
  \centering
  \caption{Examples of nodal performance measures.}
  \label{Table:PerformanceMeasures}
  \begin{tabular}{|l|l|l|}
    \cline{1-3}
    Case I  & Amount of throttled traffic & $p_i(k) := r_i(k)- x_i(k)$   \\ \cline{1-3}
    Case II  & Throttled-to-requested traffic ratio & $p_i(k) := \big(r_i(k) - x_i(k)\big)/r_i(k)$ \\ \cline{1-3}
    Case III  & Logarithm of requested-to-allowed traffic ratio & $p_i(k) := \log \big(r_i(k)/x_i(k)\big)$  \\ \cline{1-3}
    Case IV  & Amount of allowed traffic & $p_i(k) := x_i(k)$  \\ \cline{1-3}
 %   Case IV  & The Integral of the Throttled Traffic & $p_i(k) := \sum_{t=0}^k \big ( r_i(t)-x_i(t) \big)$   \\ \cline{1-3}
  \end{tabular}
\end{table*}

In this setup, 
%instead of waiting for a command from global controller (based on global aggregated information of usage),
we assume that each node updates its state $x_i(k)$ based on its neighbors' states and performance measures (\ie, a local aggregated view of usage metrics). The update law is given by the following difference equation:
\begin{equation}
  x_i(k+1)~=~x_i(k) \, +\, \gamma \sum_{i \sim j }w_{ij}\big(p_i(k) - p_j(k)\big),~~ k \in \mathbb Z_{+},
  \label{model_DST_DT}
\end{equation}
where $i \sim j$ denotes that nodes $i$ and $j$ are connected by a link in the underlying graph, $w_{ij}=w(\{i,j\})$ is the weight of link $\{i,j\}$ in graph $\mathcal G$, and parameter $\gamma$ is a positive number which depends on the size of the time step (\ie, $x(k):=x(k \Delta t)$ where $\Delta t = \gamma$). In Subsection \ref{sec:impact}, $\gamma$ is referred to as the server update cycle, and its effect on the performance analysis will be discussed.

%The continuous version (CT) of \eqref{model_DST_DT} is given by
%%
%\begin{equation}
%  d l_i(t)~=~-\sum_{j \in N(i)} a_{ij}\left ( p_j(t) - p_i(t) \right ) dt,~~ t\in \mathbb R_{++}.
%  \label{model_DST_CT}
%\end{equation}
%
The dynamics of the entire network can be written in the following compact form
\begin{equation}
  x(k+1)~=~x(k)\, + \, \gamma \, L\,  p(k),~~k\in \mathbb Z_{+},
  \label{eq:model_DT}
\end{equation}
where $x(k)$ is an $n$-by-$1$ vector of node limits at time $k$, $p(k)$ is an $n$-by-$1$ vector of nodal performance measures at time $k$, and $L$ is the Laplacian matrix of the coupling graph $\mathcal G$. {Then, the accepted number of requests at server $i$ at time $k$ is given by $a_i(k ):= \min \left \{ x_i(k), r_i(k) \right \}$.} 
The total number of requests, the total limit, and the total number of accepted requests for the entire network are defined by
\begin{equation}
r_\text{total}(k)~:=~\sum_{i=1}^{n} r_i(k),
\end{equation} 
\begin{equation}
l_\text{total} ~:=~\sum_{i=1}^{n} x_i(0),
\end{equation}
and
\begin{equation}
a_\text{total}(k)~:=~\sum_{i=1}^{n} \min \left \{ x_i(k), r_i(k) \right \},
\end{equation}
respectively.
The ideal curve for total number of accepted requests is given by 
\begin{equation}
a_\text{ideal}(k)~:=~\min \left \{ l_\text{total}, r_\text{total}(k) \right \}.
\label{a_ideal}
\end{equation}
%Therefore, the percentage of over-throttling can be defined as follows
%\[ \text{ Over-throttling \%} := \frac{\sum_{k=0}^{N} \left (a_\text{ideal}(k) - a_\text{total}(k) \right) }{\sum_{k=0}^{N} a_\text{ideal}(k)} \times 100,\]
%where $N$ is the number of cycles.

%with the following eigenvalues
%%
%begin{equation}
%  \lambda_1=0 \leq \lambda_2 \leq \cdots \leq \lambda_n.
%\end{equation}
%%
%Similarly, its CT version is given by
%%
%\begin{equation}
%  d L(t)~ =~ P(t)  dt,~~ t \in \mathbb R_{++}.
%  \label{model_CT}
%\end{equation}
%%
%\begin{remark}[Role of Underlying Graphs]
%The underlying graph shows the communication links between servers. Based on \eqref{model_DST_DT} each server updates its states based on its neighbors' states. The weights of link $\{i,j\}$ can be $0$ or $1$ (1 %if there is a communication link between them and $0$ otherwise ). Moreover, we can consider a general case where $a_{ij} \geq 0$. Therefore $a_{ij}\geq 0$ can be considered as design parameters or we can %assume the weight between two nodes is a function of the distance between the nodes (any distance which shows the ``closeness'' of two servers).
%\end{remark}

{We should note that in compact form \eqref{eq:model_DT}, weights do not disappear, and are encoded in matrix $L$. Here $L$ is the Laplacian matrix of weighted graph $\mathcal G$. Hence, off-diagonal elements of the matrix represent $-w_{ij}$'s. }

The following lemma shows that the total nodal limit is fixed over time.

\begin{lemma}
  \label{lemma_1}
The total summation of nodal limits is fixed over time, which means
\begin{equation}
   \sum_{i=1}^n x_i(k)~=~{\sum_{i=1}^n x_i(0)},~~~{\text{for all~} k \in \mathbb Z_{++}}.
  \label{eq:total_sum_limits}
\end{equation}
\end{lemma}

\begin{proof}
We multiply both sides of \eqref{eq:model_DT} by $\ones^T$ on the left and get
\begin{equation}
  \sum_{i=1}^n x_i(k+1)~=~\sum_{i=1}^n x_i(k) \,+\, \gamma \, \mathbf 1^T\, L \, p(t).
  \label{214}
\end{equation}
Assume that $p_i(k)$'s are bounded. Since $L$ is the Laplacian matrix of an undirected graph, its row and column sums are zero which, completes the proof.
\end{proof}
 
{Based on this result, the total sum of nodal limits is constant and it depends only on initial values, \ie, $l_\text{total}$.}
%\begin{remark}
A similar result is reported in \cite{Stanojevic}, which guarantees the capacity constraint for a generalized distributed rate-limiting system. The condition \eqref{214} holds for any linear consensus network even for those over directed graphs. 
%$\diamond$
%\end{remark}

In the next section, we study the overall performance of DST networks based on their nodal performance measure and the behavior of incoming network traffic.
%%%%%%%%%%%%%%%%%%%%%%%%%%%%%%%%%%%%%%%%%%%%%%%%%%%
\section{Properties of Typical Nodal Performance Measures}
\label{sec:2}

Each node $i$ is associated with a performance measure $p_i(k)$, which shows the performance of server $i$ at time $k$. %We assume that this measure is a decreasing function of $x_i$ (the limit of server $i$).
Some examples of performance measures are presented in Table \ref{Table:PerformanceMeasures}.

%positive system
%It should be noted that when $p_i(k)$ is a decreasing function of $x_i(k)$, it can shown that the model \eqref{eq:model_DT} is a positive system, \ie, that  $x(k) \geq 0$ for all $k>0$ if $x(0) \geq 0$.

In this section, we choose $p_i(k)$ to be the number of throttled requests at node $i$ at time $k$
\begin{equation}
  p_i(k)~:=~r_i(k) - x_i(k).
\end{equation}
Then, \eqref{model_DST_DT} can be rewritten as
\begin{equation}
  p(k+1)~=~\left(I-\gamma \, L\right) p(k) \,+\, \big ( r(k+1)-r(k) \big ),~~  k \in \mathbb Z_{+}.
  \label{eq-7}
\end{equation}

Based on the behavior of incoming network traffic/requests, two cases are considered.
\subsection{Steady loads}
Let us assume that the number of client requests incoming at node $i$ is constant across time:
\begin{equation}
  r_i(k+1) -r_i(k)~=~0,~~k\in \mathbb Z_{+}.
  \label{steady}
\end{equation}
%
%which means that the rate of updating is much faster than changing the behavior of signal $r_i(k)$, then we can simplify \eqref{eq-7} as below
Equation \eqref{eq-7} can then be simplified as below
\begin{equation}
  p(k+1)~=~ \left(I-\gamma \, L\right)  p(k),~~k \in \mathbb Z_{+}.
  \label{eq-9}
\end{equation}
\begin{lemma}\cite{Xiao03fastlinear} % this one is older than Olfati2007
\label{lemma2}
  For any $i,j \in \{1,2, \cdots, n\}$, we have
  \[ \lim_{k\rightarrow \infty}| p_i(k) - p_j(k)|~=~ 0,\]  if and only if $\max \{ 1- \gamma \lambda_2, \gamma \lambda_n-1\} <1$.
\end{lemma}
\begin{proof}
It is straightforward.
\end{proof}

Based on this result, as long as graph $\mathcal G$ is connected we can find a positive $\gamma$, which guarantees reaching a consensus state (for a small enough positive number $\gamma$).

%{\noindent \bf Overall Network Performance Measure:}
We can now study the convergence rate based on properties of the underlying graph and the design parameter $\gamma$.

Let us define the following performance measure which shows the convergence rate of the DST
\begin{equation}
 \Phi_{\rm cr} ~=~ \max_{i\geq 2}|1- \gamma \lambda_i|,
\end{equation}
a {smaller $ \Phi_{\rm cr}$ meaning faster asymptotic convergence.}
%Then, it is well-known that the convergence rate depends on the second smallest eigenvalue $\lambda_2$ for small $\gamma$ .

\begin{remark}[Role of Topologies for Small $\gamma$]
Networks with $n$ servers can be ranked based on their convergence rates; consequently, the path graph topology has the worst convergence rate and the complete graph has the best convergence rate (for small enough $\gamma$). Also, it can be shown that among tree graphs, star graphs have the best rate and path graphs have the worst. $\diamond$
\end{remark}

\vspace{.1cm}
\subsection{Non-Steady Loads}
Assumption \eqref{steady} is strong and can be relaxed. Let us assume that
\begin{equation}
  v_i(k+1)~:=~r_i(k  + 1 ) - r_i(k)
  \label{non-steady}
\end{equation}
where  $v(k) \in \R^n$ is a zero mean random vector such as
\begin{eqnarray}
  &&\mathbb E \left[v(k)\right]~=~\mathbf 0,\nonumber \\
  && \mathbb E \left[ v(k)v^T(k)\right]~=~{\rm Cov}(v), \nonumber \\
  && \mathbb E \left[ v(k)v^T(s)\right]~=~\mathbf 0, ~~ k \neq s.
  \label{eq:v}
\end{eqnarray}
%
%Wiener process and $\sigma^2>0$ is the variance of the noise.
Then, \eqref{eq-7} can then be simplified as below
\begin{equation}
  p(k+1)~=~ \left(I -\gamma L\right) p(k) \,+\,  v(k+1),~~ k\in \mathbb Z_{+}.
  \label{model_DST_II}
\end{equation}

%{\noindent \bf Overall Network Performance Measure:}
We can now define the following overall performance measure for the network
\begin{equation}
  \Phi_{\rm ss}~=~\lim_{t \rightarrow \infty} \mathbb E  \left[\frac{1}{2n}\sum_{i,j}\left ( p_i(k)-p_j(k)\right  )^2\right ],
\label{overall_performance}
\end{equation}
%
%\begin{equation}
%  \Phi~=~\lim_{t \rightarrow \infty} \mathbb E  \left[ P^{\rm T}(t) (I_N-\frac{1}{N}J_N) P(t)\right ],
%\label{15}
%\end{equation}
%
%where $I_N$, $J_N$ are the identity matrix and the matrix of all ones, respectively.
The quantity \eqref{overall_performance} shows the steady-state dispersion of $p_i$'s from their average \cite{SiamiTAC, SiamiTCNS, Bamieh12}.

%In the case that we have correlated and non-uniform $r_i(k+1 ) -r_i(k)$, we can use the covariance matrix of the vector $v(k)$ given by \eqref{eq:v}, we get the following closed-form formula for the overall performance of the network. 
%
%\begin{equation}
%  \Phi_{\text{ss}} =\frac{1}{2 \gamma}{\rm Trace} \left(\left(L-\frac{\gamma}{2}L^2\right)^{\dag}{\rm{Cov}}(v)\right),
%\end{equation}
%
%where is the Moore–Penrose pseudoinvers [ i.e., $A^{\dag}=\left(A + \frac{1}{n}J_n\right)^{-1}-\frac{1}{n}J_n$ , where $J_n$ is a $n$-by-$n$ matrix of all ones] and ${\rm{Cov}}(v)$ is the covariance matrix of random vector $v(k)$. 
The following theorem presents a closed-form formula for the overall performance of DST \eqref{model_DST_II}, based on the Laplacian matrix of the underlying graph and the covariance matrix of the input vector $v$.

\begin{theorem}
  \label{th-1}
  For a given DST \eqref{model_DST_II}, the overall performance measure \eqref{overall_performance} can be quantified as
\begin{equation}
  \Phi_{\text{ss}} =\frac{1}{2 \gamma}{\rm Trace} \left[\left(L-\frac{\gamma}{2}L^2\right)^{\dag}{\rm{Cov}}(v)\right],
\end{equation}
where ${\rm{Cov}}(v)$ is the covariance matrix of random vector $v(k)$. 
\end{theorem}

\begin{proof}
  The overall performance measure is the same as the squared $\mathcal H_2$-norm of a discrete linear time invariant system \eqref{model_DST_II}. Therefore, the measure can be quantified as follows:
  \begin{equation}
    \Phi_{\rm ss}~=~ {\rm Trace}\left [ Q \, {\rm Cov}(v) \right],
    \label{eq:trace}
  \end{equation}
  where $Q \succeq \mathbf 0 $ is the solution of the following discrete Lyapunov equation
  \[ (I -\gamma L) Q (I -\gamma L)^{T} ~-~ Q~ +~ \big({ I - 1/n \mathbf 1 \mathbf 1^T}\big) ~=~\mathbf 0.\]
  By doing some calculation it follows that
  \begin{equation}
    Q ~= ~ {(2\gamma L-\gamma^2 L^2)^{\dag}}.
    \label{matrix_P}
  \end{equation}
 Using \eqref{eq:trace} and \eqref{matrix_P} we get the desired result. 
\end{proof}

\begin{remark}[Independent $v_i$'s]
In the case where $v_i$'s are independent then ${\rm Cov}(v)$ is a diagonal matrix 
$\gamma\diag(\sigma_1^2,\ldots,\sigma_n^2)$ where $\sigma_i$ depends on the property of signal $r_i$. We get

\begin{equation}
  \Phi_{\text{ss}}~=~\frac{1}{2} \sum_{i=1}^n  c_{ii}^{\dag}\sigma^2_{i},
  \label{decompose_formula}
\end{equation}
where $(L- \frac{\gamma}{2} L^2)^{\dag}=[c_{ij}^{\dag}]$. Based on \eqref{decompose_formula}, we can obtain a centrality measure for servers. Indeed, $c_{ii}^{\dag}$ shows the impact of server $i$ on the overall performance. 
See \cite{SiamiTCNS} for more details on centrality measures with respect to $\mathcal H_2$-norm of the system (the focus of \cite{SiamiTCNS} is on the class of continuous-time linear consensus networks however.) $\diamond$
\end{remark}

\begin{remark}[Independent and identically-distributed $v_i$'s]
Based on Theorem \ref{th-1}, the overall performance measure of the network can be calculated based on spectral eigenvalues of the coupling graph and the variance of changing demands (i.e., $r_i(k+1) -r_i(k) \sim  N(0 ,\gamma \, \sigma^2)$) as follows
\begin{equation}
  \Phi_{\rm ss}\,=\,
  \begin{cases}
    \sum_{i=2}^n \frac{\sigma^2}{\lambda_i(2-\gamma \, \lambda_i)}, ~0 < \lambda_i < 2/\gamma~~\text{for}~ i=2, \cdots, n,\\
    \infty, ~~~~~~~~~~~~~~~~~\,\text{otherwise}.
    \end{cases}
  \label{eq:overall_performance}
\end{equation}
We note that condition $0 < \lambda_i < 2/\gamma$ for $i=2, \cdots, n$ is the same as the one needed for the system without noise to converge (cf., Lemma \ref{lemma2}). $\diamond$
\end{remark}
The quantity \eqref{eq:overall_performance} has a close connection with the ``total effective resistance''
of an electric network as follows
\begin{equation}
  \lim_{\gamma \rightarrow 0} \Phi_{\rm ss} ~=~ \frac{\sigma^2}{2n}\sum_{i>j}r_{ij},
\end{equation}
where $r_{ij}$ is the effective resistance between node $i$ and $j$, \ie,
\[ r_{ij} ~:=~ l_{ii}^{\dag}~+~l_{jj}^{\dag}~-~l_{ij}^{\dag}~-~l^{\dag}_{ji}, ~~ L^{\dag}=[l_{ij}^{\dag}]. \]
For more details see \cite{Ghosh}.

%\begin{remark}[Limitations and drawbacks]
%Note that $r_i(k)$ and $x_i(k)$ are nonnegative numbers. Therefore, in some cases where $r_i(k)$'s are small we could have an additional constraint on $x_i(k)$ to make it nonnegative. 
%\end{remark}

\begin{remark} [Another interpretation of the overall measure] Let us assume that $r_i(0)$'s are given with the normal distribution, and $r_i$'s remain constant. Then the expected total mismatch loss can be obtained based on 
\begin{eqnarray}
  &&\mathbb E \left[\frac{1}{n} \sum_{k=0}^\infty\sum_{i>j}\left ( p_i(k)-p_j(k)\right  )^2 \Delta t \right]~=~ \nonumber \\
  &&~~~~~\frac{1}{2 \gamma}{\rm Trace} \left[\left(L-\frac{\gamma}{2}L^2\right)^{\dag}{\rm{Cov}}(v)\right] ~=~ \Phi_{\text{ss}}.
\end{eqnarray}
$\diamond$
\end{remark}

\vspace{.1cm}
Due to space limitations, other nodal performance measures defined in Table \ref{Table:PerformanceMeasures} are briefly analyzed in the appendix.
\section{DST Optimization Problems}
\subsection{Impact of the Server Update Cycle}
\label{sec:impact}
{In this subsection, we study the effect of the server update cycle $\gamma$ on our analysis. As shown in Section \ref{sec:2}, the overall performance measure of a DST depends on its Laplacian eigenvalues and the server update cycle.
  To enhance the overall performance of the network, one can obtain the optimal update cycle for all servers.

The following theorem presents the optimal update cycle for a DST in the case of steady loads (\ie, when the number of client requests is constant across time).
\vspace{.1cm}
  \begin{theorem}
  \label{th:2}
    For a given DST \eqref{eq-9} with a graph $\mathcal G$, the optimal update cycle is given by
    \begin{equation}
      \gamma_{\rm optimal}~=~\frac{2}{\lambda_2 + \lambda_n}.
    \end{equation}
    \end{theorem}
\vspace{.1cm}
  \begin{proof} We need to solve the following convex optimization
\begin{eqnarray}
  &&\underset{\gamma > 0}{\rm minimize}~~\max_{i \geq 2} \left |1- \gamma \lambda_i \right | .
\end{eqnarray}
It is not difficult to see that $2(\lambda_2 + \lambda_n)^{-1}$ minimizes the cost function. We have 
\[ 0 < \lambda_2 \leq \cdots \leq \lambda_n, \]
and, accordingly, we can rewrite the cost function as follows
\begin{eqnarray}
  &&\max_{i \geq 2} |1- \gamma \lambda_i|  = \max \left \{ 1 - \gamma \lambda_2, \gamma \lambda_n -1\right \}.
  \label{550}
\end{eqnarray}
To minimize \eqref{550}, we need
\begin{equation}
1 - \gamma \lambda_2 ~=~\gamma \lambda_n -1,
\label{556}
\end{equation}
since if $1 - \gamma \lambda_2 \neq \gamma \lambda_n -1$,  one can decrease the cost function by increasing or decreasing $\gamma$ . Therefore, the optimal $\gamma$ is the solution of \eqref{556}. This completes the proof.
  \end{proof}

  In the case of non-steady loads, having a closed-form formula for the optimal update time based on the Laplacian eigenvalues seems difficult. However, one can obtain the solution by solving the following convex optimization problem:
\begin{eqnarray}
  &&\underset{\gamma > 0}{\rm minimize}~~ \frac{1}{2\gamma} {\rm Trace} \left[\left(L-\frac{\gamma}{2}L^2\right)^{\dag}{\rm{Cov}}(v)\right].
\end{eqnarray}
{In the case where $v_i$'s are independent and identically-distributed (\ie, ${\rm{Cov}}(v)=\gamma\diag(\sigma^2,\ldots,\sigma^2)$), the optimal update time can be bounded from above and below by $1/\lambda_2$ and $1/\lambda_n$, respectively.}

\subsection{DST Synthesis Problems}
\label{sec:synthesis}
In this subsection, we present our main results on the design of optimal distributed rate-limiting systems. We formulate our problems as convex optimization problems. The questions we are trying to answer in this section are
\begin{itemize}
\item[-] What are the optimal link weights for the {\it fastest} DST network?
\item[-] What are the optimal link weights for the {\it most robust} DST network?
%\item What is the optimal topology of the network for the fastest when some constraints are imposed on the network topology?
%\item Where is the best location in the network for adding a new server?
%\item Where should we add a communication link to improve the overall performance of the system?
%\item How does the overall performance scale when the number of servers increases?   
\end{itemize}

Depending on which nodal and overall performance measures are chosen, one can come up with different optimal topologies.

\subsubsection{The Fastest DST Process}

Here we briefly describe the problem of finding the fastest DST on a given underlying topology, {where `fastest' means the one with the smallest $ \Phi_{\rm cr}$.}
The optimal weights can be found by solving the following optimization problem
\begin{eqnarray}
  &&\underset{w(e)}{\rm minimize}~~ \max_{i \geq 2} |1- \gamma \lambda_i|   \label{optimal_fast} \\
  &&{\rm subject~to}~~~ w(e) ~\geq~ 0,~~{\rm for~ all}~~ e \in E.  \nonumber  
\end{eqnarray}
This problem was studied before in \cite{Boyd2006}. Problem \eqref{optimal_fast} can be cast as a semidefinite programming (SDP) problem as follows
\begin{eqnarray}
  &&\underset{w(e), \theta}{\rm minimize}~~  \theta \\
  &&{\rm subject~to}~-\theta I \preceq  I - \gamma \sum_{e \in E}w(e)L_e - \frac{1}{n}\mathbf1 \mathbf 1 ^{T} \preceq \theta I, \nonumber \\
  &&~~~~~~~~~~~~~w(e) \geq 0,~~ e \in E,  \nonumber  
\end{eqnarray}
where $L_e$ is the unweighted Laplacian matrix of link $e$.

\subsubsection{The Most Robust DST Process}

Here we briefly describe the problem of finding the most robust DST on a given underlying topology, {where `most robust' means the one with the smallest $ \Phi_{\rm ss}$. }
The optimal weights can be found by solving the following problem
\begin{eqnarray}
  &&\underset{w(e)}{\rm minimize}~~ \frac{1}{2\gamma} {\rm Trace} \left[\left(L-\frac{\gamma}{2}L^2\right)^{\dag}{{\rm{Cov}}(v)}\right]     \label{optimal_robust} \\
  &&{\rm subject~to}~~~ w(e) ~\geq~ 0,~~{\rm for~ all}~~ e \in E, \nonumber \\
  &&~~~~~~~~~~~~~~~~~L~=~\sum_{e \in E}w(e)L_e, \nonumber \\ 
    &&~~~~~~~~~~~~~~~~\max_{i \geq 2} |1- \gamma \lambda_i| ~\leq~  1.  \nonumber  
\end{eqnarray}
We note that $\Phi_{\rm ss} = \frac{1}{2\gamma} {\rm Trace} \left(\left(L-\frac{\gamma}{2}L^2\right)^{\dag}{{\rm{Cov}}(v)}\right)$ is a convex function of the link weights. To find the solution of \eqref{optimal_robust} one can use a variety of standard methods for convex optimization (\eg, interior-point methods and subgradient-based methods).

\vspace{.1cm}
\begin{theorem}
Problem \eqref{optimal_robust} can be formulated as a SDP problem as follows
\begin{eqnarray}
  &&\underset{w(e), Y}{\rm minimize}~~ \frac{1}{2\gamma} {\rm Trace} \left[ Y {{\rm{Cov}}(v)}\right]  - \frac{\mathbf 1^{T} {{\rm{Cov}}(v)} \mathbf 1}{2n\gamma^2}\nonumber \\
  &&{\rm subject~to}~~~ w(e) ~\geq~ 0,~~{\rm for~ all}~~ e \in E, \nonumber  \\
  &&~~~~~~~~~~~~~~~~L~=~\sum_{e \in E}w(e)L_e, \nonumber  \\
  &&~~~~~~~~~~~~~~~~\mathbf 0~ \preceq~I -  \frac{1}{2}\left ( \gamma L + (1/n)\mathbf 1 \mathbf 1^T\right)~ \preceq~ I,  \nonumber  \\
  &&~~~~~~~~~~~~~~~\,\begin{bmatrix}
    L + \frac{1}{\gamma n}\mathbf 1 \mathbf 1 ^{T}& L & I\\
    L  & \frac{2}{\gamma }I  & \mathbf 0 \\
    I& \mathbf 0 & Y
  \end{bmatrix} ~\succeq~ 0. 
  \label{Big_matrix}
\end{eqnarray}
\end{theorem}
\vspace{.1cm}
\begin{proof}
  We need the following condition to hold in order to guarantee that the network is marginally stable:
\begin{equation}
\mathbf 0~ \preceq~I -  \frac{1}{2}\left (\gamma L + (1/n)\mathbf 1 \mathbf 1^T\right)~ \preceq~ I. 
\label{stability}
\end{equation}  
 Then, according to \eqref{stability} and the Schur complement condition for positive semidefiniteness it follows that
\begin{equation}
	\begin{bmatrix}
	    	L +\frac{1}{\gamma n}\mathbf 1 \mathbf 1 ^{T}&  L \\
    		 L  & \frac{2}{\gamma} I   
  	\end{bmatrix} ~\succeq~ 0. 
	\label{b-matrix}
\end{equation}
  Again, using the Schur complement condition for positive semidefiniteness, \eqref{Big_matrix} and \eqref{b-matrix}, we get the following equivalent condition
\begin{equation}
  Y - \frac{1}{\gamma n} \mathbf 1 \mathbf 1 ^{T}~ \succeq~ \left (L - \frac{\gamma}{2} L^2 \right)^{\dag},
\end{equation}
which completes the proof. 
\end{proof}

%
%This Problem can be cast as a SDP problem as follows
%
%\begin{eqnarray}
%  &&{\rm minimize}~~ {\rm Trace} (Y)  \nonumber  \\
%  &&{\rm subject~to}~~~~ w(e) \geq 0,~~{\rm for~ all}~~ e \in E \nonumber \\
%  &&~~~~~~~~~~~~~~~~\sum_{e \in E} w(e) \leq M \nonumber \\
%  &&~~~~~~~~~a
%\end{eqnarray}

%%%%%%%%%%%%%%%%%%%%%%%%%%%%%%%%%%%%%%%%%%%%%%%%%%%

%\begin{remark}%[Uniting a Local Controller With a Global Controller]
%  We suggest an algorithm that combines a local controller with a global controller. A local controller has access to local information but it is faster that a global controller. On the other hand, a global controller has more information but is slow. This leads us to the idea of uniting local and global controllers.   To do so, we use a hybrid feedback controller.
%\end{remark}

%\begin{figure}
%  \centering
%  \includegraphics[trim = 5 5 5 5, clip,width=0.4\textwidth]{figure/gst_tree.pdf}
%  \caption{A GST with three layers and 10 servers.}
%\end{figure}
%

%
%\begin{figure}
%  \centering
%  \includegraphics[trim = 5 5 5 5, clip,width=0.4\textwidth]{figure/GST_one_level.pdf}
%  \caption{A GST with one layers and $n$ servers.}
%  \label{GST}
%\end{figure}

%%%%%%%%%%%%%%%%%%%%%%%%%%%%%%%%%%%%%%%%%%%%%%%%%%%

\section{Illustrative Numerical Simulations}
\label{sec:numerical}
In this section, we support our theoretical developments with illustrative examples that provide better insight into the role of the underlying graph topology in the DST network.
\begin{figure}
  \centering
  \includegraphics[trim = 5 5 5 5, clip,width=0.5\textwidth]{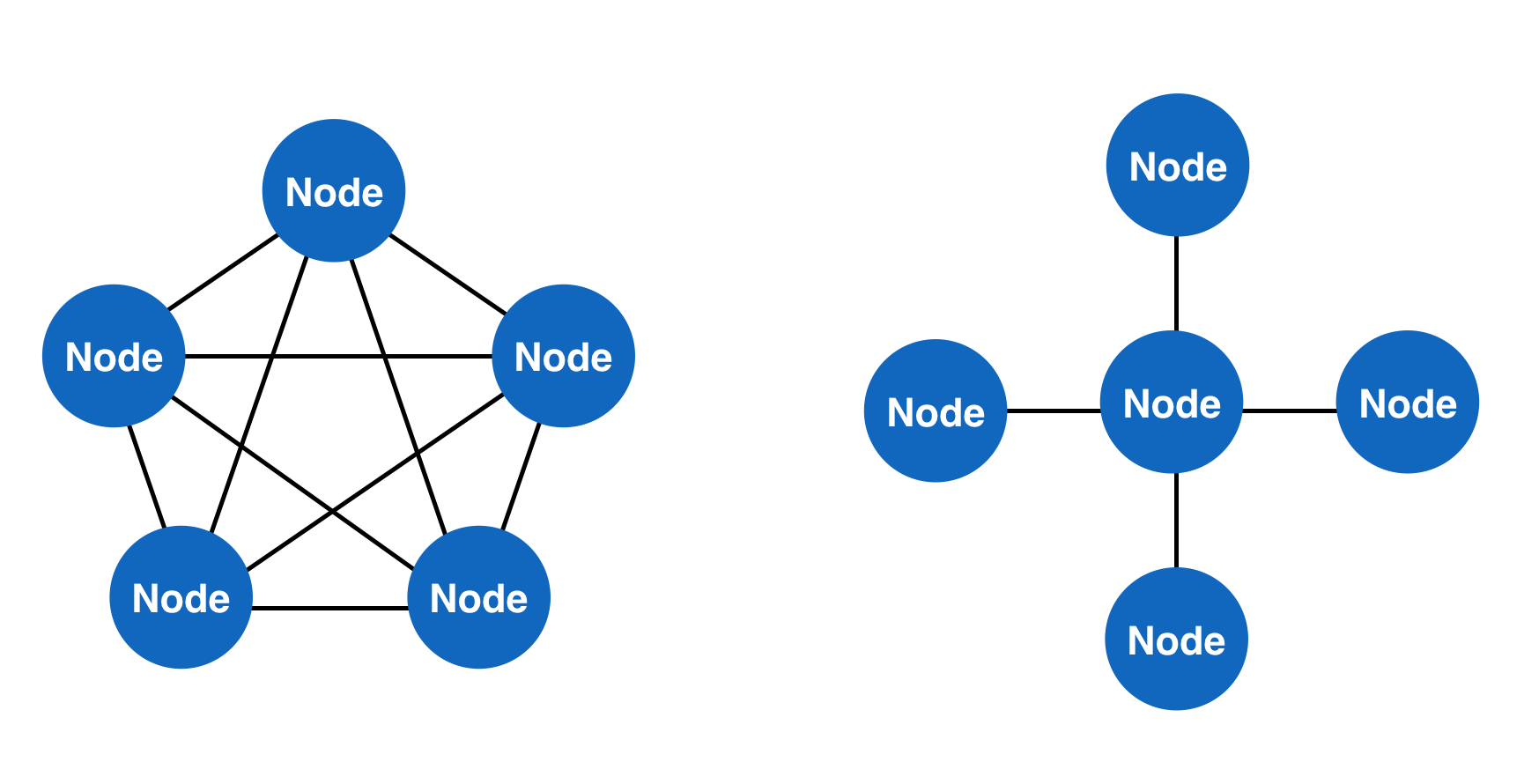}
  \caption{ Two DST networks with five servers over a complete graph and a star graph.}
  \label{exam-1-0}
\end{figure}
\vspace{.1cm}
\begin{example}
  Consider two DST networks with five servers over complete graph $\mathcal K_5$ and star graph $\mathcal S_5$ as depicted in Fig. \ref{exam-1-0}. Let us assume that the update cycle is given and fixed (without loss of generality $\gamma = 1 $). Based on the results presented in Theorem \ref{th:2}, one can obtain the optimal weight links for both networks to get the fastest DST (See Table II).
\begin{table}[h]
  \centering
  \caption{Optimal link weights.}
  \label{tab}
  \begin{tabular}{|c|c|c|}
    \cline{1-3}
     & Complete Graph $\mathcal K_5$ & Star Graph $\mathcal S_5$   \\ \hline \hline
   Optimal Weight & $w(e)=1/5$ & $w(e)=1/3$  
    \\ \cline{1-3}
  \end{tabular}
\end{table}

For each network the weights are uniform since their underlying graphs are edge-transitive. 
\end{example}

\vspace{.1cm}
\begin{example}
 Consider two DST networks with $10$ servers over graphs depicted in Figs. \ref{exam-1-1} and \ref{exam-1-2}. Let us assume that in both graphs all links have a weight of one. Based on the results presented in Theorem \ref{th:2}, one can obtain the optimal update cycle for both networks to get the fastest DST (see Table III).
\begin{table}[h]
  \centering
  \caption{Optimal update cycles.}
  \label{tab}
  \begin{tabular}{|c|c|c|}
    \cline{1-3}
     & Graph $\#1$ & Graph $\#2$   \\ \hline \hline
   Optimal update cycle & $ \Delta t = 0.4226 $ & $\Delta t=0.2222$  
    \\ \cline{1-3}
  \end{tabular}
\end{table}

Moreover, let us consider $1,000$ clients that are randomly assigned to $10$ servers such that each server has $100$ clients. 
  Fig. \ref{exam} shows the simulation results that are obtained for each of these two DST networks given a randomly generated usage curve over $1,000$ cycles. As expected, the overall performance of the DST over graph $\#2$ is better (\ie, over-throttling is less severe) than the performance of the DST over graph $\#1$ (for small time step $\gamma = 0.02$). 
  
 { 
\begin{table}[h]
  \centering
  \caption{Overall network performance measures.}
  \label{tab}
  \begin{tabular}{|c|c|c|}
    \cline{1-3}
     & Graph $\#1$ & Graph $\#2$   \\ \hline \hline
  $\Phi_\mathbf{cr}$ & $ 0.9969 $ & $0.9727$  \\
    \cline{1-3}
    $\Phi_\mathbf{ss}$ & $ 334.7965 $ & $69.3075$  \\
     \cline{1-3}
    Over-throttling $\%$ & $6.2 \%$ & $2.8 \%$  \\
     \cline{1-3}
  \end{tabular}
\end{table}
%  
  
%In Figs. \ref{exam-1-1} and \ref{exam-1-2}, each blue bar show the average allowed requests from each client, and red bar shows the average amount of throttled requests for each client.  

\begin{figure}
  \centering
  \includegraphics[trim = 5 5 5 5, clip,width=0.5\textwidth]{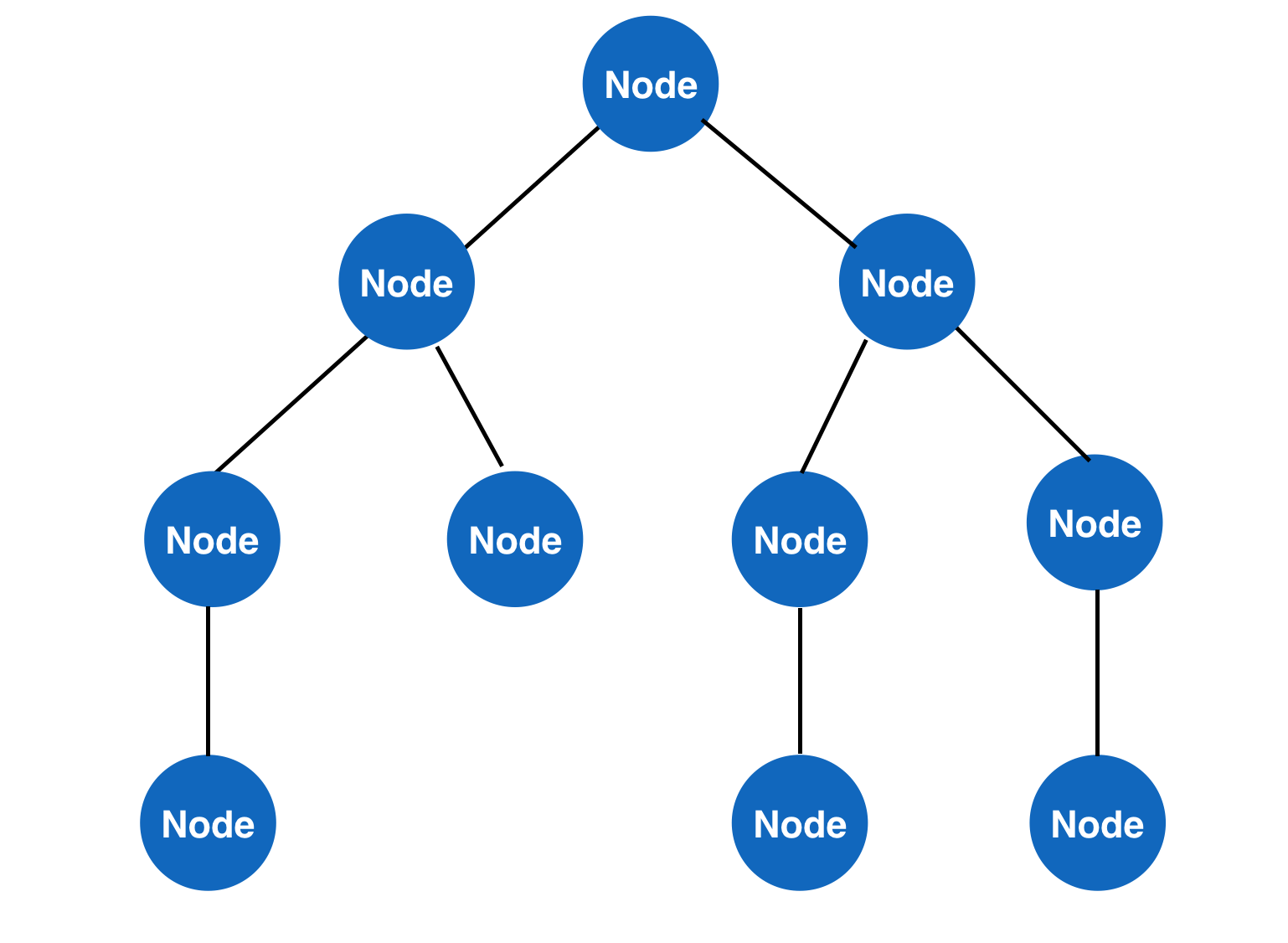}
  \caption{ A DST network with $10$ servers over a tree graph (graph $\#1$).}
  \label{exam-1-1}
\end{figure}

\begin{figure}
  \centering
  \includegraphics[trim = 5 5 5 5, clip,width=0.48 \textwidth]{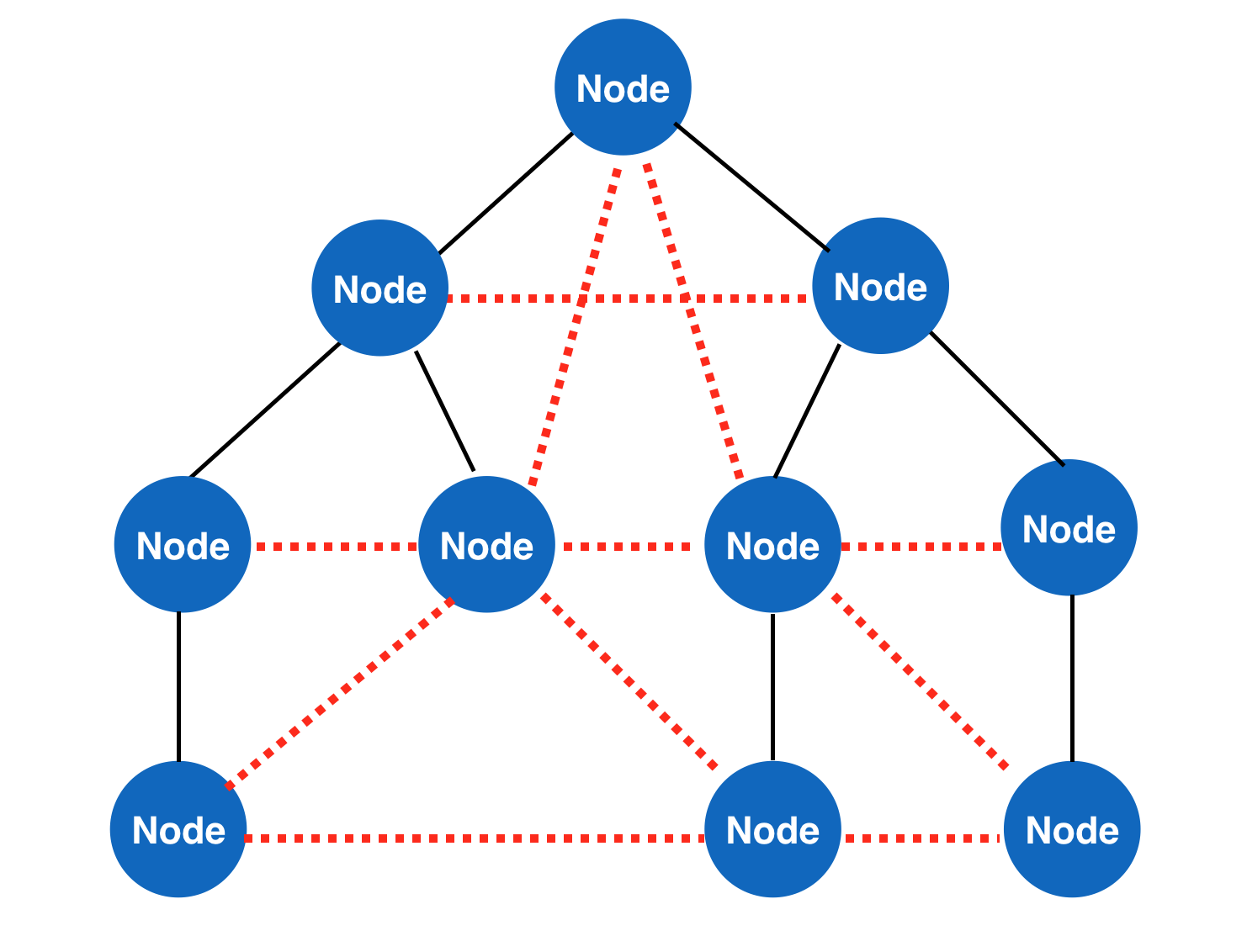}
  \caption{ A DST network with $10$ servers over a tree graph with some additional red dotted links (graph $\#2$).}
  \label{exam-1-2}
\end{figure}

\begin{figure}[t]
  \centering
 \includegraphics[width=.9\textwidth]{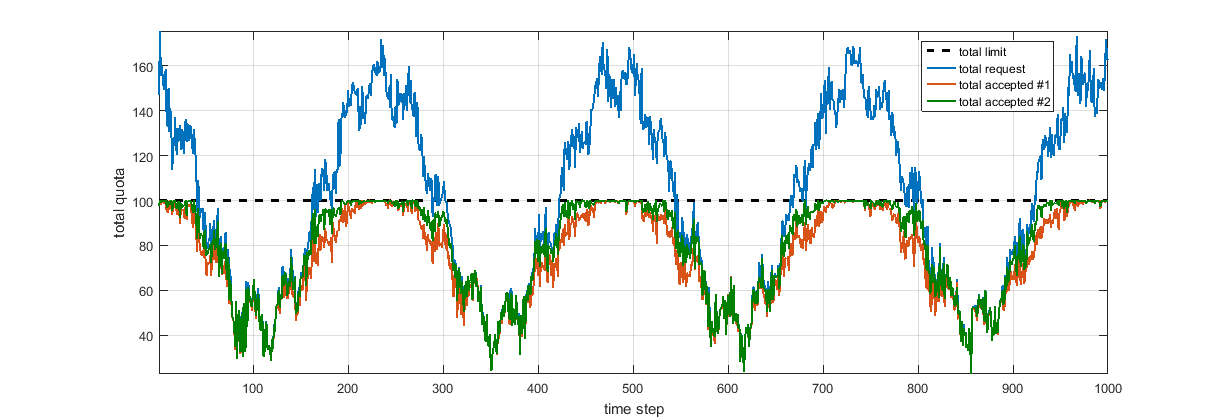}
  \caption{The simulation results for two DST networks with $\gamma = 0.02 $ over graphs $\#1$ and $\#2$ with $10$ nodes. }
  \label{exam}
\end{figure}

In Fig. \ref{exam}, the blue curve shows the total number of requests versus time, \ie, {$r_\text{total}(k)$, }
the black dashed line presents the total limit for the entire network, \ie, {$l_\text{total}$,}
and the red and green curves show the total accepted requests for graph $\#1$ and graph $\#2$ respectively, \ie, {$a_\text{total}(k)$. }
We should note that, the ideal curve for total accepted requests is given by \eqref{a_ideal}.
Therefore, the percentage of over-throttling can be defined as follows
\[ \text{ Over-throttling \%} := \frac{\sum_{k=0}^{N} \left (a_\text{ideal}(k) - a_\text{total}(k) \right) }{\sum_{k=0}^{N} a_\text{ideal}(k)} \times 100,\]
where $N$ is the number of cycles (in this example $1,000$).}

%  Consider a DST \eqref{model_DST_II} with $10$ servers over a complete graph with $\gamma = 1 $ and optimal weight links, and a global throttler with $10$ leaf servers as shown in Fig. \ref{GST}. Assume that each server has performance measure $p_i(k)= r_i(k)-x_i(k)$. The trottling algorithm which is used in the global throttler is to ....  the In the case of steady loads the global throttler works better. However for non-steady loads DST performs pretty well. 

%\begin{figure}
%  \centering
%%  \includegraphics[trim = 100 230 100 230, clip,width=0.5\textwidth]{figure/k-5.pdf}
% \includegraphics[width=0.5\textwidth]{figure/fig_1.jpg}
%  \caption{The simulation result for the DST with $\gamma = 0.02 $ over graph $\#1$ with 10 nodes.}
%  \label{exam-1-1}
%\end{figure}
%%
%
%%
%\begin{figure}
%  \centering
%%  \includegraphics[trim = 100 230 100 230, clip,width=0.5\textwidth]{figure/s-5.pdf}
% \includegraphics[width=0.5\textwidth]{figure/fig_2.jpg}
%  \caption{The simulation result for the DST with $\gamma = 0.02 $ over graph $\#2$ with $10$ nodes. }
%  \label{exam-1-2}
%\end{figure} 
\end{example}
%%%%%%%%%%%%%%%%%%%%%%%%%%%%%%%%%%%%%%%%%%%%%%%%%%%
{\section{Throttling Algorithms at the Node Level}
\label{sec:algorithms}

In this part, we focus on the structure of each node. Each node consists of a server with its clients (see, for example, Fig. \ref{Fig_micro}). Besides update law \eqref{model_DST_DT}, it has its own throttling algorithm to handle its clients. 
Let us assume that server $i$ has $c_i$ clients, and $r_i^{(j)}(k)$ is the number of requests received  by server $i$ from its $j$-th client at time $k$. Therefore, the total number of client requests received by server $i$ at time $k$ is 
\[r_i(k)=\sum_{j=1}^{c_i}r_i^{(j)}(k).\] Let's define $x_i^{(j)}(k)$ as a limit on the number of requests of $j$-th client of server $i$ that is allowed to service at time $k$. 
%Thus, we have 
%\[x_i(k)=\sum_{j=1}^{c_i}x_i^{(j)}(k),\] which is the total limit on the number of allowed requests for server $i$. Note that $x_i(k)$ updates based on dynamics \eqref{eq:model_DT}.  %%This is not correct! and depends on nodal algorithms
In each update cycle, first each server (let's say server $i$) collects all metrics from its clients (\ie, number of requests $r_i^{(j)}(k)$ for $j \in \{1,2, \cdots, c_i \}$) as well as collecting its neighbors' states and performance measures (\ie, a local aggregated view of usage metrics), then aggregates all metrics and updates its state, and finally pushes new limits to its clients (\ie, $x_i^{(j)}(k)$ for $j \in \{1,2, \cdots, c_i \}$). {It also communicates its local aggregated view of usage metrics to its neighboring nodes' servers. }

\begin{figure}[t]
  \centering
  \includegraphics[trim = 5 5 5 5, clip,width=0.4\textwidth]{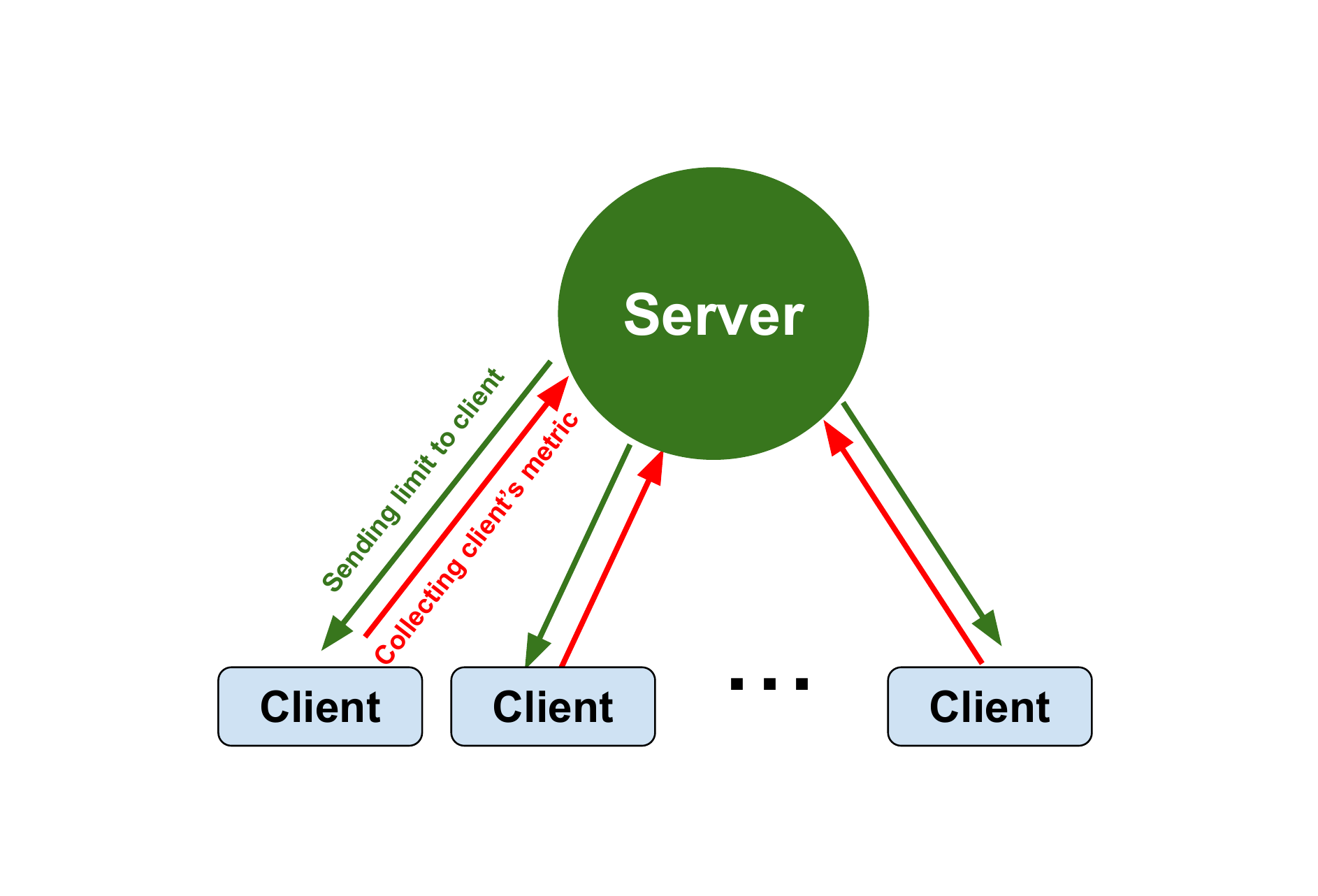}
  \caption{A server with its clients.}
  \label{Fig_micro}
\end{figure}

At the node level, viable throttling algorithms can be considered to throttle same amount, ratio, or the logarithm of ratio from all tasks until the total limit is reached (please see nodal performance measures in Table I). In what follows, we present two simple throttling algorithms with their high-level examples which can be used by each node.

{The first algorithm keeps the throttled ratios uniform over all tasks and is defined in Algorithm 1.}
\begin{algorithm}[t]
    \SetKwInOut{Input}{Input}
    \SetKwInOut{Output}{Output}

%    \Input{$r_i^{(j)}(k)$: client requests for $j \in \{1,2, \cdots, c_i \}$ at time $k$, and $x_i(k)$: the total allowed limit to service at time $k$.}
%           \vspace{.2cm}
%   \Output{$x_i^{(j)}(k)$: limits on the number of requests for $j \in \{1,2, \cdots, c_i \}$ at time $k$ }
 %   		\vspace{.4cm}
{
    \Input{$r_i^{(j)}(k)$ for $j \in \{1,2, \cdots, c_i \}$ and $x_i(k)$.}
            \vspace{.2cm}
    \Output{$x_i^{(j)}(k)$ for $j \in \{1,2, \cdots, c_i \}$. }
    		\vspace{.4cm}
}
   
    $r_{i} (k) := \sum_{j=1}^{c_i} r_i^{(j)}(k)$\\
    \eIf{$r_i (k) \leq l_i(k)$}
      {
        % for all $j \in \{1,2, \cdots, c_i \}$, let $x_i^{(j)}(k) := r_i^{(j)}(k)$
        \For{$j = 1$ {\it to} $c_i$}
        {$x_i^{(j)}(k) := r_i^{(j)}(k)$}
      }
      {
        % for all $j \in \{1,2,\cdots,c_i \}$, let $x_i^{(j)}(k) :=  \frac{x_i(k)}{r_{i} (k)} r_i^{(j)}(k)$
        \For{$j = 1$ {\it to} $c_i$}
        {$x_i^{(j)}(k) :=  \frac{x_i(k)}{r_{i} (k)} r_i^{(j)}(k)$}
      }
    \caption{A simple balancing algorithm for keeping throttled ratios uniform at server $i$ {at time $k$}}
\end{algorithm}
The second algorithm demonstrates a simple load balancing algorithm which distributes incoming requests across all tasks as uniformly as possible by throttling large number of requests. We present the steps of this algorithm in Algorithm 2.

 \begin{algorithm}[t]
    \SetKwInOut{Input}{Input}
    \SetKwInOut{Output}{Output}

%    \Input{$r_i^{(j)}(k)$: client requests for $j \in \{1,2, \cdots, c_i \}$ at time $k$, and $x_i(k)$: the total allowed limit to service at time $k$.}
%	    \vspace{.2cm}
%    \Output{$x_i^{(j)}(k)$: limits on the number of requests for $j \in \{1,2, \cdots, c_i \}$ at time $k$\\}
%    	\vspace{.4cm}
    \Input{$r_i^{(j)}(k)$ for $j \in \{1,2, \cdots, c_i \}$ and $x_i(k)$.}
	    \vspace{.2cm}
    \Output{$x_i^{(j)}(k)$ for $j \in \{1,2, \cdots, c_i \}$}
    	\vspace{.4cm}

   $r_{i} (k) := \sum_{j=1}^{c_i} r_i^{(j)}(k)$\\
   
   \eIf{$r_i (k) \leq l_i(k)$}
      {
       % for all $j \in \{1,2,\cdots,c_i \}$, let $x_i^{(j)}(k) := r_i^{(j)}(k)$
        \For{$j = 1$ {\it to} $c_i$}
        {$x_i^{(j)}(k) := r_i^{(j)}(k)$}
      }
      {
        {sort} $r_i^{(j)}(k)$ for $j \in \{1,2, \cdots, c_i \}$ $\rightarrow$ $r_i^{(j)\uparrow}(k)$\\
        $s:=0$\\
        $l:= \frac{x_i(k)}{c_i}$\\
        \For{$j = 1$ {\it to} $c_i$}
        {
        \If{$l > r_i^{(j)\uparrow}(k)$}
        {
        $s=s-r_i^{(j)\uparrow}(k) +l$\\
	$l=\frac{s}{c_i-j}+l$
        }
        }
        %for all $j \in \{1,2, \cdots, c_i \}$, let $x_i^{(j)}(k) := l$
        \For{$j = 1$ {\it to} $c_i$}
        {$x_i^{(j)}(k) := l$}
      }
    \caption{A simple load balancing algorithm with throttling large number of requests at server $i$ at time $k$}
\end{algorithm}

%For considering different algorithms on servers, we could use different performance measures $p_i$ for nodes. 
%
%\begin{figure}
%  \centering
%  \includegraphics[trim = 175 175 175 175, clip,width=0.45\textwidth]{figure/squash_algorithm.pdf}
%  \caption{Requested {\color{red}quota} and throttled {\color{red}quota} for each user based on a simple load balancing algorithm. The blue dashed line shows the resulting allowed level $l$ in Algorithm 2.}
%  \label{fig-squash}
%\end{figure}
%

%requested {\color{red} quota} (\ie, number of requests)

We now present two high-level examples according to these algorithms. Figs. \ref{fig-algorithm}-a  and \ref{fig-algorithm}-b  depict numbers of requests and throttled requests for server $i$ based on Algorithms 1 and 2, respectively.  Each bar shows the number of requests per client. Blue bars show clients' requests. The red area shows the throttled request traffic. The clients are sorted by by number of requests in ascending order. The blue dashed line in Fig. \ref{fig-algorithm}-b shows the allowed limit on each task. We should note that the total number of request at this server (server $i$) is the area of all bars, \ie, $r_i(k)$, and the total allowed request is the area of all blue bars, \ie, $a_i(k)$.

%%%%%%%%%%%%%%%%%%%%%%%%%%%%%%%%%%%%%%%%%%%%%%%%%%%

\section{Conclusion}
\label{sec:conclusion}
In this paper, we investigated performance deterioration (\eg, over-throttling) of distributed system throttlers with respect to external uncertainties and server time cycles. We developed a graph-theoretic framework to relate the underlying structure of the system to its overall performance measure. We then compared the performance/robustness of the proposed distributed system throttlers with different underlying graphs. A promising research direction is to investigate the overall performance measure of DST networks with respect to the other nodal performance measures.

\begin{figure}
\centering
  \includegraphics[trim = 175 175 175 175, clip,width=0.45\textwidth]{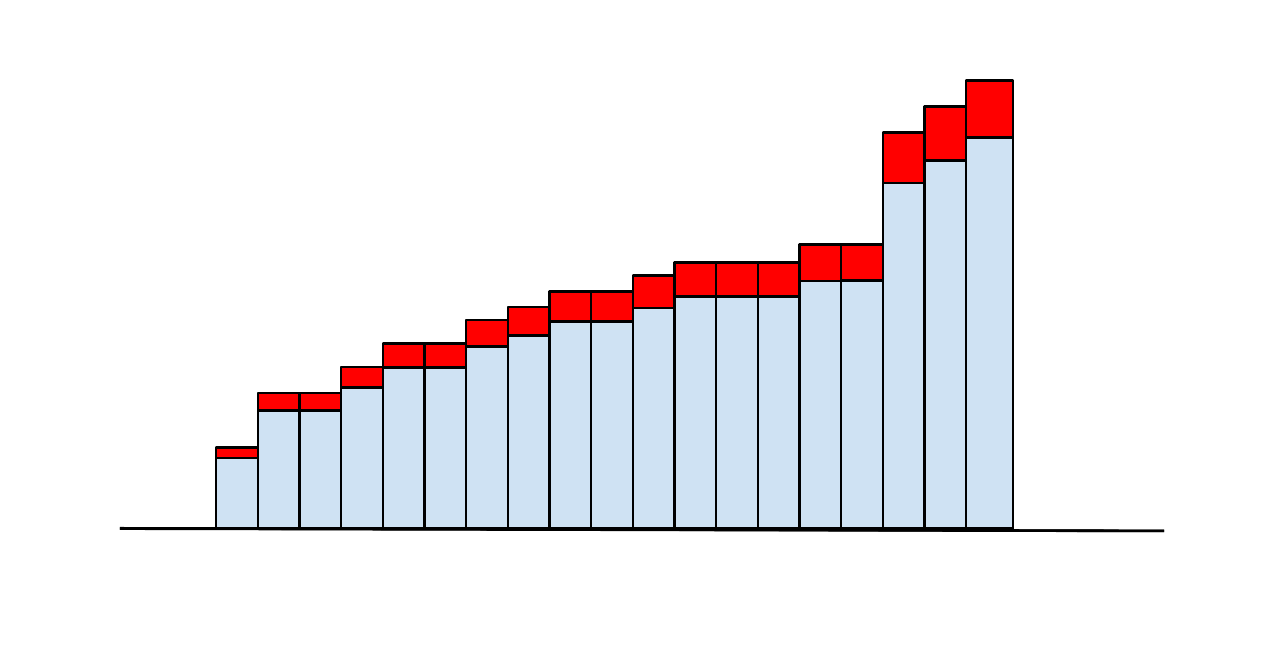}
  \caption*{(a)}
  \includegraphics[trim = 175 175 175 175, clip,width=0.45\textwidth]{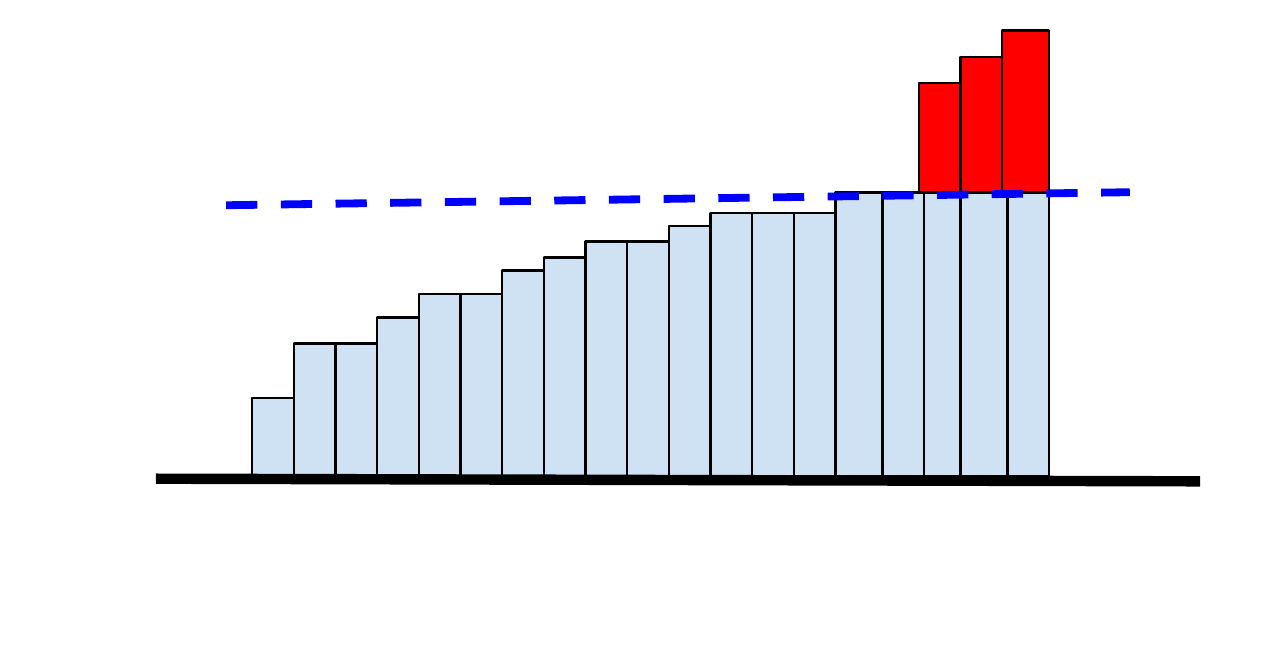}
  \caption*{(b)}
\caption{Requested quota (\ie, number of requests) and throttled requests for server $i$: (a) Algorithm 1 which keeps the throttled ratios uniform over all tasks, and (b) Algorithm 2 which throttles large number of requests. The blue dashed line shows the resulting allowed level $l$ in Algorithm 2.}
\label{fig-algorithm}
\end{figure}

%%%%%%%%%%%%%%%%%%%%%%%%%%%%%%%%%%%%%%%%%%%%%%%%%%%
\section*{Acknowledgment}
The authors would like to thank Dr. Sergei Vassilvitskii and Dr. Hossein Azari for their valuable comments and suggestions to improve the paper.

%%%%%%%%%%%%%%%%%%%%%%%%%%%%%%%%%%%%%%%%%%%%%%%%%%%
\appendix[Other Nodal Performance Measures]
\label{sec:appendix}
In this part, we present the dynamics of the DST for other nodal performance measure defined in Table I  ({\it Case I} is studied in Section \ref{sec:2}).

\subsubsection*{ Case II} Assume that the performance measure at server $i$ is given by
\begin{equation}
  p_i(k)~=~\frac{ r_i(k) - x_i(k)}{r_i(k)},
\end{equation}
and $r_i(k) >0$.
Then, we can rewrite \eqref{model_DST_DT} in the following form
\begin{eqnarray}
&& \hspace{-.6cm}  x(k+1)=\gamma \, L \, \diag\left[{r_1(k)^{-1}}, \cdots , {r_n(k)^{-1}}\right] \left( r(k)- x(k) \right) \nonumber \\
  &&~~+~ x(k),~~k\in \mathbb Z_{+}.  
\end{eqnarray}
Let assume that $r_i(k) = \mathfrak r$ for all $k \in \Z_+$ and $i \in \left\{1,2,\cdots,n \right \}$. So, we have
\begin{equation}
  x_i(k)~=~- \mathfrak r \left (p_i(k) -1 \right).
\end{equation}
Then, it follows that
\begin{equation}
p(k+1)~=~\left (I - \frac{\gamma}{\mathfrak r} L \right ) p(k).
\end{equation}
%
%\begin{eqnarray}
% &&\diag \left [ r_1(k)+\epsilon, \cdots , r_n(k)+\epsilon \right ] (p(k+1)-p(k)) ~=~\nonumber \\
% &&~~~~-\gamma \, L\, p(k),   k \in \mathbb Z_{++}.
%\end{eqnarray}
%
In this case, in addition to the coupling graph and the update cycle, the values of $\mathfrak r$ plays a role in the convergence rate of the network (same for other overall performance measures).

\subsubsection*{Case III} Next, we assume that the performance measure at server $i$ is given by
\begin{equation}
  p_i(k)~=~\log r_i(k)- \log x_i(k).
\end{equation}
Assume that $r_i(k) = \mathfrak r$, for all $k \in \Z_+$ and $i \in \left\{1,2,\cdots,n \right \}$. Therefore, we get
\begin{equation}
  x_i(k)~=~\mathfrak r \, e^{-p_i(k)}.
\end{equation}
Then, it follows that
\begin{equation}
\exp \left(-p(k+1) \right)~ =~ \exp \left(-p(k) \right) + \frac{\gamma}{\mathfrak r} \, L\,  p(k),~~k\in \mathbb Z_{+},
\label{947}
\end{equation}
where $\exp p(k):= \left [ e^{p_1(k)}, \cdots, e^{p_n(k)} \right ]^T$. Let us define 
\begin{equation} 
\bar p(k) := \exp \left(-p(k) \right),
\label{951}
\end{equation}
using \eqref{947} and \eqref{951}, it follows that
\begin{equation}
\bar p(k+1)~ =~ \bar p(k)-\frac{\gamma}{\mathfrak r} \, L\,  \ln \bar p(k),~~k\in \mathbb Z_{+},
\label{952}
\end{equation}
where 
\[\ln \bar p(k)~:=~ \big [ \ln {\bar p_1(k)}, \cdots, \ln {\bar p_n(k)} \big ]^T.\]
%\begin{eqnarray}
%  &&\diag \left[r_1(k)e^{-p_1(k)}, \cdots,  r_n(k) e^{-p_n(k)}\right ] \left ( p(k+1) - p(k) \right) \nonumber \\
% &&~~~~~ =~ -\gamma \, L\,  p(k),~~k\in \mathbb Z_{++}.
%\end{eqnarray}
%

\subsubsection*{Case IV} Finally, let us assume that 
\[p_i(k)~=~x_i(k),\]
for $i=1, \cdots, n$. Then, dynamics \eqref{model_DST_DT} can be rewritten in the following form
%n the steady state we get the limits which are obtained based on a simple load balancing algorithm (where $x_i$'s are nonnegative or we put a constraint on their signs). Consider the following dynamics 
%
\begin{equation}
 p(k+1)~=~\left (I + \gamma L \right)p(k).
\end{equation}
In this case, based on Lemma \ref{lemma2} the system is unstable, which means the state trajectories are unbounded. Therefore we consider additional constraints to make them bounded as follows: the state of node $i$ at time $k+1$ is not updated (\ie, $x_i(k+1) ~=~ x_i(k)$) and its information at time $k$ is not used for updating the states of neighboring nodes at time $k+1$ when 
\begin{itemize}
\item[-] $x_i(k) = r_i(k)$ and $x_i(k+1) - x_i(k) > 0$,
\item[-] $x_i(k) = 0$ and $x_i(k+1)-x_i(k) < 0$.
\end{itemize}
We should note that also in this case the following equality holds
\[ \sum_{i=1}^nx_i(k)~=~ \sum_{i=1}^n x_i(0).\]
In a steady-state, each state $x_i$ reaches its boundaries (i.e., $0$ and $r_i$) or a value between them.

\begin{spacing}{1.2}
\bibliography{ref/main_Milad}
\end{spacing}

\end{document}